\documentclass[british]{article}
\usepackage{ae,aecompl}
\usepackage[T1]{fontenc}
\usepackage[latin9]{inputenc}
\usepackage{parskip}
\usepackage{colortbl}
\usepackage{mathtools}
\usepackage{amsmath}
\usepackage{amsthm}
\usepackage{amssymb}
\usepackage{geometry}
\makeatletter

\providecommand{\tabularnewline}{\\}

\numberwithin{equation}{section}

\usepackage{capt-of}
\usepackage{bbm}
\date{}
\makeatother

\theoremstyle{plain}
\newtheorem{thm}{\protect\theoremname}[section]
\theoremstyle{remark}
\newtheorem{notation}[thm]{\protect\notationname}
\theoremstyle{definition}
\newtheorem{defn}[thm]{\protect\definitionname}
\theoremstyle{remark}
\newtheorem{rem}[thm]{\protect\remarkname}
\theoremstyle{plain}
\newtheorem{prop}[thm]{\protect\propositionname}
\newtheorem{lem}[thm]{\protect\lemmaname}
\usepackage{babel}
\providecommand{\definitionname}{Definition}
\providecommand{\lemmaname}{Lemma}
\providecommand{\notationname}{Notation}
\providecommand{\propositionname}{Proposition}
\providecommand{\remarkname}{Remark}
\providecommand{\theoremname}{Theorem}

\begin{document}
\global\long\def\IN{\mathbb{N}}%
\global\long\def\II{\mathbbm{1}}%
\global\long\def\IZ{\mathbb{Z}}%
\global\long\def\IQ{\mathbb{Q}}%
\global\long\def\IR{\mathbb{R}}%
\global\long\def\IC{\mathbb{C}}%
\global\long\def\IP{\mathbb{P}}%
\global\long\def\IE{\mathbb{E}}%
\global\long\def\IV{\mathbb{V}}%
\global\long\def\artanh{\textup{artanh}\,}%
\global\long\def\corr{\textup{corr}}%

\title{Reconstructing the Probability Measure of a Multi-group Curie-Weiss
Model with Interacting Groups}
\author{Miguel Ballesteros\thanks{IIMAS-UNAM, Mexico City, Mexico}, Gerardo
Franco Cordova\footnotemark[1], Fedro Guillen\footnotemark[1], and
Gabor Toth\footnotemark[1]
\thanks{Corresponding author, e-mail: gabor.toth@iimas.unam.mx}}
\maketitle
\begin{abstract}
\noindent We study the problem of reconstructing the probability
measure of a multi-group version of the Curie-Weiss or mean-field
model of ferromagnetism from a sample of the voting behaviour of the
population. While originally used to study phase transitions in statistical
mechanics, the Curie-Weiss model has been applied to study phenomena
where many agents interact with each other, in particular in case
of a heterogeneous population with identifiable subpopulations. The
degree of social cohesion within social groups manifests in the way
the members of the group influence each others' decisions as well
as how they behave under outside influence from voters belonging to
another group. While the maximum likelihood estimator of the coupling
parameters has desirable statistical properties in theory, computational
challenges make applications to larger populations impractical. Therefore,
we analyse an estimator based on asymptotic approximations to the
behaviour of the Curie-Weiss model valid for large populations. Due
to the wide application of models such as Curie-Weiss, the estimator
is potentially useful in disciplines such as political science, sociology,
automated voting, and preference aggregation.
\end{abstract}


\section{Introduction}

The investigation of ferromagnetic systems has long served as a cornerstone
of statistical mechanics, providing a rigorous framework for understanding
how interactions between large collections of microscopic particles
generate macroscopic order. Among the most influential mathematical
models in this area are the Ising and Curie-Weiss models, both of
which combine conceptual simplicity with rich mathematical behaviour
and extensive applicability. Introduced by Wilhelm Lenz and Ernst
Ising in the 1920s (see \cite{Ising1925}), the Ising model assigns
binary spin variables to points on the lattice $\IZ^{d}$, where $d\in\IN$,
and assumes that each spin is influenced by its nearest neighbours
and by an external magnetic field. Although the one-dimensional (i.e.,
$d=1$) model admits an exact solution and does not undergo a finite-temperature
phase transition, higher-dimensional versions exhibit far more intricate
phenomena, including critical behaviour and phase transitions, whose
existence was first established rigorously by Onsager for the two-dimensional
case in \cite{Onsager1944}.

The Curie-Weiss model, proposed by Pierre Curie and Pierre Weiss in
the early 20th century in order to investigate ferromagnetic phase
transitions, replaces local interactions with an all-to-all interaction
structure in which every spin interacts equally with every other spin.
This mean-field approximation preserves fundamental collective effects,
most notably spontaneous magnetisation, while substantially simplifying
the mathematical analysis. Beyond its origins in statistical physics,
the model has become a versatile framework in numerous disciplines,
including neuroscience \cite{BouMcDMu2007}, economics, where mean-field
models of strategic interaction are widely studied \cite{JovaRose1988,BD2001},
queueing theory \cite{BaKaKPRS1992}, sociology \cite{CGh2007,AlBaKhBu2022},
and political science \cite{Ki2007,To2020phd,DieKauKa2023}. Its broad
appeal stems from its ability to represent phenomena such as consensus
formation, social influence, and political polarisation using only
a small number of interaction parameters.

A natural extension of the Curie-Weiss model arises in settings where
the population cannot reasonably be regarded as homogeneous. In such
cases, the population is partitioned into $M\in\IN$ distinct groups,
each of which may exhibit different interaction characteristics. The
resulting multi-group Curie-Weiss model allows interaction strengths
to vary between groups, thereby capturing heterogeneous social or
physical environments. This framework has attracted increasing attention
over the years (see, for example, \cite{FedeCont2011,KnLoScSi2020,KirsToth2022b})
and has been applied to problems including opinion dynamics, voting
theory \cite{KT2021c}, and statistical community detection \cite{BRS2019,LoweSchu2020,BaMePeTo2023}.
The interaction between individuals is governed by the so-called coupling
parameters: higher values of the coupling parameters result in stronger
interaction, which means that voters pay more attention to the opinions
of their fellow voters, thus giving rise to higher levels of social
cohesion, ideological alignment, or cultural similarity across subpopulations.

The problem of statistical estimation of coupling parameters has been
studied for several decades, beginning with the classical work \cite{Kastelei1956}
and continuing with more recent studies such as \cite{GalBarCo2009}
and \cite{FedVerCo2013}. These contributions considered estimation
based on observations from the entire population. More recently, a
series of articles has focused on the problem of estimating the coupling
parameters of a Curie-Weiss model with several independent groups:
\cite{BaMeSiTo2025} established that the maximum likelihood estimator
is consistent, asymptotically normal, and satisfies a large deviation
principle, thereby laying a theoretical foundation for the the use
of the maximum likelihood estimator in this context. Nevertheless,
its practical implementation is hindered by the need to evaluate the
partition function, whose computational complexity grows exponentially
with the population size. To overcome this limitation, \cite{BalNauTo2025}
introduced a computationally efficient approximation to the maximum
likelihood estimator that is valid in the large-population limit.
This estimator can be computed from complete voting data with constant
computational complexity and retains the principal asymptotic properties
of the maximum likelihood estimator, although consistency requires
the population to be sufficiently large. Despite its computational
advantages, the estimator still relies on observations from the full
population. Finally, \cite{BalNauTo2025a} analyses estimators which
can be calculated from observing subsets of votes as small as two
voters per group at the cost of higher variance of the estimators
in some cases (see the discussion in \cite[Section 3.1]{BalNauTo2025a}).

The present article analyses the multi-group Curie-Weiss model with
two groups and correlated voting across group boundaries, thus removing
the constraint to independent groups found in previous work. The model
can be found in Definition \ref{def:CWM}. This model has three coupling
parameters: $J_{1,1}$, which controls interactions between voters
belonging to group 1, $J_{2,2}$, which has the same function for
group 2, and $J_{1,2}$ which determines the interactions between
voters belonging to different groups. Thus, in the simplest case of
two interacting groups of voters, we have one additional parameter
to estimate compared to the independent group case. As we will see,
this added complexity leads to some limitations on the estimation
of the coupling parameters in the low temperature regime of the model
(see Definition \ref{def:regimes}), which corresponds to strong interactions
between voters.

Allowing interactions between voters belonging to different groups
in the model treated in this article has some advantages and some
disadvantages. The foremost advantage is the possibility to apply
the estimator to situations in which strong cross-border influences
are evident and a model neglecting those would miss a large part of
the picture. Specifically, a more general setting with interactions
between voters belonging to different groups may shed light on situations
in which a large and influential group affects a smaller group's decisions
more than the interactions within the smaller group itself. There
are two main drawbacks of the general model:
\begin{enumerate}
\item The independent group model with $M$ groups treated in \cite{BaMeSiTo2025,BalNauTo2025,BalNauTo2025a}
has the minimum number of coupling parameters $\beta_{\lambda}$ necessary
to estimate the social cohesion in each group: $M$ coupling parameters,
one for the interactions within each of the groups. The more general
model in this article, on the other hand, has a total of $\frac{M\left(M+1\right)}{2}$
parameters, thus requiring more observations to estimate all coupling
parameters. The coupling parameters can also be estimated from separate
votes concerning different issues when the groups are assumed to be
independent. This makes applying the model with independent groups
far more practical when data availability is limited.
\item In the independent group model, there is a single coupling parameter
that controls interactions within a given group. Therefore, we have
a single numerical value which we can interpret as `social cohesion'
within that group. This makes comparisons between groups simple. The
general model with $\frac{M\left(M+1\right)}{2}$ coupling parameters
exhibits a complicated dependence structure where the correlation
between votes is a function of the whole set of coupling parameters.
For example, the correlation between the votes of individual 1 from
group 1 and individual 1 from group 2 in general depends on the entire
set of coupling parameters, and not just on a single interaction parameter
linking groups 1 and 2. The same applies to interactions within each
group. Thus it is less straightforward to interpret the coupling parameters
and their corresponding estimates.
\end{enumerate}
Taking into account these points, the independent group model certainly
has its advantages and is the practical choice in a large number of
situations.

Section \ref{sec:model} presents the multi-group model whose coupling
parameters are the subject of this study. Section \ref{sec:Sampling}
discusses assumptions which underlie the sampling of observations.
The estimators we define and analyse can be found in Section \ref{sec:Estimation},
Definitions \ref{def:high_temp_est} and \ref{def:low_temp_est}. Our
theoretical results are presented in Section \ref{sec:Results}, with
Theorem \ref{thm:limit} summing up the statistical properties of
the estimators. Section \ref{sec:Simulations} discusses simulation
results for the estimators with the goal of giving a panorama of other
properties of the estimators, including the frequencies of detecting
the regimes of the model correctly, frequencies of the estimators
falling within certain tolerances from the true parameter values,
and the frequency of correctly detecting the dependence between the
two groups. The latter question is particularly important in light
of the discussion of the advantages and disadvantages of the independent
vs.\! the dependent group models above.

This work lies at the intersection of statistical mechanics, statistical
inference, and institutional design. By extending inference methods
to the case of correlated voting across group boundaries using the
multi-group Curie-Weiss models, it contributes to an expanding body
of research that applies ideas from statistical physics to questions
of collective decision-making, representation, and fairness. At the
same time, it provides rigorous methodological tools for analysing
and designing voting systems that reflect empirically observed patterns
of group behaviour, complemented by numerical simulations, thus illustrating
the applicability of the estimators to groups of fairly small size
all the way to large countries. The presentation is intended to be
accessible to researchers from a range of disciplines, including mathematics,
statistics, physics, economics, and political science, by combining
mathematical precision with conceptual exposition.

\section{\label{sec:model}The Curie-Weiss Model}

Assume there are two groups of voters of size $N_{1}\in\IN$ and $N_{2}\in\IN$,
respectively. Let $N\coloneq N_{1}+N_{2}$ be the size of the overall
population. Each of the $N$ voters is indexed by the group $\lambda\in\IN_{2}$
and by $i\in\IN_{N_{\lambda}}$\footnote{We will write $\IN_{m}\coloneq\left\{ 1,\ldots,m\right\} $ for any
$m\in\IN$.}, with the vote being a random variable $X_{\lambda,i}$ taking values
in $\left\{ -1,1\right\} $. The voters take into account the opinions
by all other voters in both groups. This mutual influence is represented
mathematically by a coupling matrix.

We define the coupling matrix $J$ to be a positive definite $2\times2$
matrix
\begin{equation}
J\coloneq\left(J_{\lambda,\mu}\right)_{\lambda,\mu=1,2}\in\IR^{2\times2}.\label{eq:J}
\end{equation}

\begin{notation}
For any symmetric matrix $A$ with real entries, we will write $A>0$
for the fact that $A$ is positive definite and $A\geq0$ if $A$
is positive semi-definite.
\end{notation}

For any voting configuration $\left(x_{1,1},\ldots,x_{1,N_{1}},x_{2,1},\ldots,x_{2,N_{2}}\right)\in\left\{ -1,1\right\} ^{N}$,
the Hamiltonian $\mathbb{H}:\left\{ -1,1\right\} ^{N}\rightarrow\IR$
is given by 
\begin{equation}
\mathbb{H}\left(x_{11},\ldots,x_{1N_{1}},x_{21},\ldots,x_{2N_{2}}\right)\coloneq-\frac{1}{2}\sum^{2}_{\lambda,\mu=1}\frac{J_{\lambda,\mu}}{\sqrt{N_{\lambda}N_{\mu}}}\sum^{N_{\lambda}}_{i=1}\sum^{N_{\mu}}_{j=1}x_{\lambda,i}x_{\mu,j}.\label{eq:Hamiltonian}
\end{equation}

\begin{defn}
\label{def:CWM}Let $J\in\IR^{2\times2}$ with $J>0$ and $\mathbb{H}$
as defined in (\ref{eq:Hamiltonian}). The multi-group Curie-Weiss
model's probability measure $\mathbb{P}$, which gives the probability
of each of the $2^{N}$ voting configurations, is defined by 
\begin{align}
\mathbb{P}_{J,N_{1},N_{2}}\left(X_{1,1}=x_{1,1},\ldots,X_{2,N_{2}}=x_{2,N_{2}}\right) & :=Z^{-1}_{J,N_{1},N_{2}}e^{-\mathbb{H}\left(x_{11},\ldots,x_{2N_{2}}\right)}\label{eq:CWM}
\end{align}
for all $\left(x_{1},\ldots,x_{N}\right)\in\left\{ -1,1\right\} ^{N}$,
where $Z_{J,N_{1},N_{2}}$ is a normalisation constant which depends
on $N_{1}$, $N_{2}$, and $J$.
\end{defn}

Each group's collective decision is represented by the sum of all
votes, the so-called voting margin.
\begin{defn}
\label{def:voting_margins}For each group $\lambda\in\IN_{2}$, the
group voting margin $S_{\lambda}$ is defined by
\[
S_{\lambda}\coloneq\sum^{N_{\lambda}}_{i=1}X_{\lambda,i}.
\]
\end{defn}

\begin{rem}
\label{rem:indep_groups}The two groups are independent, i.e.\! the
marginal distributions of $\left(X_{1,1},\ldots,X_{1,N_{1}}\right)$
and $\left(X_{2,1},\ldots,X_{2,N_{2}}\right)$ are independent, if
and only if $J_{1,2}=0$. If $J_{1,2}=0$, then the correlation $\corr\left(X_{\lambda,i},X_{\lambda,j}\right)$
between any two votes within each group $\lambda\in\IN_{2}$ depends
only on the coupling constant $J_{\lambda,\lambda}$. However, if
$J_{1,2}\neq0$, and hence the groups are not independent, $\corr\left(X_{\lambda,i},X_{\mu,j}\right)$
will depend on the whole set of coupling constants $J_{1,1},J_{2,2}$,
and $J_{1,2}$. This means in particular, if we fix the values of
$J_{1,2}$ and $J_{2,2}$, then $\corr\left(X_{1,i},X_{2,j}\right)$
is higher for $J'_{1,1}$ than for $J_{1,1}$ if $J'_{1,1}>J{}_{1,1}$.
Cf. Lemma \ref{lem:indep_corr}.
\end{rem}

\begin{notation}
Throughout this article, we will use the symbol $\IE X$ for the expectation
and $\IV X$ for the variance of some random variable $X$. Capital
letters such as $X$ will denote random variables, while lower case
letters such as $x$ will denote realisations of the corresponding
random variable. We will use the symbol $I$ for the identity matrix
whose dimension should be obvious from the context.
\end{notation}

We will be working with two regimes of the Curie-Weiss model:
\begin{defn}
\label{def:regimes}If $I-J>0$ holds, we say that the model is in
the\emph{ high temperature regime}. If $I-J\geq0$ does not hold,
then we say the model is in the \emph{low temperature regime}.
\end{defn}

\begin{rem}
In the high temperature regime, voters interact more loosely with
each other, the conditions are more chaotic with a tendency to ties
in the groups. The low temperature regime, on the other hand, is more
ordered and public opinion tends to clearly favour one of the two
options.
\end{rem}

\section{\label{sec:Sampling}Sampling}

Let $\Omega\coloneq\left\{ -1,1\right\} $ be the space of binary
choice each voter faces in each election or referendum. Then a single
vote of the entire population including both groups is an element
of
\[
\Omega_{N_{1}+N_{2}}\coloneq\left\{ -1,1\right\} ^{N_{1}+N_{2}}.
\]
We assume we have access to a sample of $n\in\IN$ i.i.d.\! observations
from the distribution in Definition \ref{def:CWM}. Any such sample
takes values in the sampling space
\[
\Omega^{n}_{N_{1}+N_{2}}\coloneq\prod^{n}_{i=1}\Omega_{N_{1}+N_{2}}.
\]
We define the statistic which will allow us to estimate the coupling
matrix $J$.
\begin{defn}
\label{def:T_stat}We define the statistic $T:\Omega^{n}_{N_{1}+N_{2}}\rightarrow\IR^{2\times2}$
for any realisation of the sample $x=\left(x^{(1)},\ldots,x^{(n)}\right)\in\Omega^{n}_{N_{1}+N_{2}}$
by
\[
T\left(x\right)\coloneq\frac{1}{n}\sum^{n}_{t=1}\left(\begin{array}{cc}
\left(\sum^{N_{1}}_{i=1}x^{(t)}_{1,i}\right)^{2} & \sum^{N_{1}}_{i=1}x^{(t)}_{1,i}\sum^{N_{2}}_{j=1}x^{(t)}_{2,j}\\
\sum^{N_{1}}_{i=1}x^{(t)}_{1,i}\sum^{N_{2}}_{j=1}x^{(t)}_{2,j} & \left(\sum^{N_{2}}_{j=1}x^{(t)}_{2,j}\right)^{2}
\end{array}\right).
\]
\end{defn}

\begin{defn}
\label{def:J_tilde}Assume $I-J>0$ holds. For all $N_{1},N_{2}\in\IN$,
we define $\tilde{J}_{N_{1},N_{2}}\in\IR^{2\times2}$ by the implicit
condition
\[
\left(I-\tilde{J}_{N_{1},N_{2}}\right)^{-1}=\left(\begin{array}{cc}
\IE_{J,N_{1},N_{2}}\frac{S^{2}_{1}}{N_{1}} & \IE_{J,N_{1},N_{2}}\frac{S_{1}S_{2}}{\sqrt{N_{1}N_{2}}}\\
\IE_{J,N_{1},N_{2}}\frac{S_{1}S_{2}}{\sqrt{N_{1}N_{2}}} & \IE_{J,N_{1},N_{2}}\frac{S^{2}_{2}}{N_{2}}
\end{array}\right).
\]
Now suppose $I-J\geq0$ not hold. For all $N_{1},N_{2}\in\IN$, we
define $\tilde{m}_{N_{1},N_{2}}\in\IR^{2\times2}$ by the implicit
condition
\[
\tilde{m}_{N_{1},N_{2}}=\left(\begin{array}{cc}
\IE_{J,N_{1},N_{2}}\left(\frac{S_{1}}{N_{1}}\right)^{2} & \IE_{J,N_{1},N_{2}}\frac{S_{1}S_{2}}{N_{1}N_{2}}\\
\IE_{J,N_{1},N_{2}}\frac{S_{1}S_{2}}{N_{1}N_{2}} & \IE_{J,N_{1},N_{2}}\left(\frac{S_{2}}{N_{2}}\right)^{2}
\end{array}\right).
\]
\end{defn}

The above definition is motivated by the following limit result.
\begin{prop}
\label{prop:large_pop_approx}If $I-J>0$ holds, then we have, for
all $\lambda,\nu\in\IN_{2}$,
\[
\IE_{J,N_{1},N_{2}}\frac{S_{\lambda}S_{\nu}}{\sqrt{N_{\lambda}N_{\nu}}}\xrightarrow[N_{\lambda},N_{\nu}\rightarrow\infty]{}\left(\left(I-J\right)^{-1}\right)_{\lambda,\nu}.
\]
 If $I-J\geq0$ does not hold and $J_{1,2}\neq0$ \footnote{If $J_{1,2}=0$, the two groups are independent and we can estimate
the coupling constants $J_{1,1}$, $J_{2,2}$ separately. See the
articles \cite{BaMeSiTo2025,BalNauTo2025,BalNauTo2025a} which treat
the estimation problem for independent groups.}, then there is a unique point $\left(m_{1},m_{2}\right)\in\IR^{2}$
with $m_{1}>0$ and $m_{2}\neq0$ such that
\[
\IE_{J,N_{1},N_{2}}\frac{S_{\lambda}S_{\nu}}{N_{\lambda}N_{\nu}}\xrightarrow[N_{\lambda},N_{\nu}\rightarrow\infty]{}m_{\lambda}m_{\nu}.
\]
The sign of $m_{2}$ is the same as that of $J_{1,2}$.
\end{prop}

This result was first proved in \cite{KirsToth2020,KirsToth2020b}
for two very similar models. We present the sketch of a proof in Section
\ref{subsec:Sketch} for the readers' convenience.

By Proposition \ref{prop:large_pop_approx}, in the high temperature
regime $\tilde{J}_{N_{1},N_{2}}$ approximates the coupling matrix
$J$ well, provided the group populations $N_{1}$ and $N_{2}$ are
large. In the low temperature regime, we do not have a limit directly
in terms of $J$, but the magnitude and sign of $m_{1},m_{2}$ allow
certain conclusions regarding the coupling constants: the larger in
magnitude the coupling constants $J_{1,1},J_{2,2},\left|J_{1,2}\right|,$
the larger the limits $m_{1},\left|m_{2}\right|$ as $N_{1},N_{2}\rightarrow\infty$.
$m_{1}$ and $m_{2}$ each depend on all three values $J_{1,1},J_{2,2},J_{1,2}$,
so that in general a larger value of $m_{1}$ is not necessarily due
to $J_{1,1}$ being larger, it could also be due to a larger value
of $J_{2,2}$ or $J_{1,2}$. The sign of $m_{2}$ is the same as that
of $J_{1,2}$. Hence, if the voters in different groups are negatively
coupled, we will see this reflected as a negative sign of $m_{2}$
and a negative limit $m_{1}m_{2}$ of $\IE_{J,N_{1},N_{2}}\frac{S_{1}S_{2}}{N_{1}N_{2}}$
as $N_{1},N_{2}\rightarrow\infty$.

\section{\label{sec:Estimation}Estimation of $J$}

The maximum likelihood estimator $\hat{J}$ for the coupling matrix
$J$ satisfies ---and is defined by--- the optimality condition
\[
\IE_{\hat{J}\left(x\right),N_{1},N_{2}}S_{\lambda}S_{\nu}=\left(T\left(x\right)\right)_{\lambda,\nu}
\]
for all $\lambda,\nu\in\IN_{2}$ and all samples $x\in\Omega^{n}_{N_{1}+N_{2}}$.
Instead of using the exact expectation on the left hand side above,
we use the large-population approximation provided in Proposition
\ref{prop:large_pop_approx}.

The estimator we will define is thus based on Proposition \ref{prop:large_pop_approx}
and Definition \ref{def:J_tilde}. Due to the differing magnitudes
of the expectations in the two regimes, we define two separate estimators.
Which one is applicable depends on the sample.
\begin{defn}
\label{def:high_temp_est}Let $N_{1},N_{2}\in\IN$ and $J>0$. We
set
\[
\boldsymbol{N}\coloneq\left(\begin{array}{cc}
\sqrt{N_{1}} & 0\\
0 & \sqrt{N_{2}}
\end{array}\right).
\]
The estimator $\hat{J}_{N_{1},N_{2},n}:\Omega^{n}_{N_{1}+N_{2}}\rightarrow\IR^{2\times2}\cup\left\{ \textup{\ensuremath{\ell}}\right\} $
is defined as follows: for all $x\in\Omega^{n}_{N_{1}+N_{2}}$ such
that $T\left(x\right)$ is invertible and $I-\boldsymbol{N}T\left(x\right)^{-1}\boldsymbol{N}>0$
holds, we set 
\[
\hat{J}_{N_{1},N_{2},n}\left(x\right)=I-\boldsymbol{N}T\left(x\right)^{-1}\boldsymbol{N}.
\]
For all other $x\in\Omega^{n}_{N_{1}+N_{2}}$, we set $\hat{J}_{N_{1},N_{2},n}\left(x\right)\coloneq\ell$.
\end{defn}

Above, we use the symbol $\ell$ to indicate that the information
in a sample does not allow us to estimate a coupling matrix in the
high temperature regime and we therefore conclude that it is likely
the coupling matrix $J$ belongs to the low temperature regime. In
that case, we can estimate the concentration point $\left(m_{1},m_{2}\right)$
in the following way:
\begin{defn}
\label{def:low_temp_est}Let $N_{1},N_{2}\in\IN$ and $J>0$. The
estimator $\hat{m}_{N_{1},N_{2},n}:\Omega^{n}_{N_{1}+N_{2}}\rightarrow\IR^{2\times2}$
is defined for all $x\in\Omega^{n}_{N_{1}+N_{2}}$ such that $T\left(x\right)$
is singular or $I-\boldsymbol{N}T\left(x\right)^{-1}\boldsymbol{N}>0$
does not hold by
\[
\hat{m}_{N_{1},N_{2},n}\left(x\right)\coloneq\boldsymbol{N}^{-1}\boldsymbol{N}^{-1}T\left(x\right)\boldsymbol{N}^{-1}\boldsymbol{N}^{-1}.
\]
For all other $x\in\Omega^{n}_{N_{1}+N_{2}}$, we set $\hat{m}_{N_{1},N_{2},n}\left(x\right)\coloneq0$.
\end{defn}

\begin{rem}
Note that by Proposition \ref{prop:large_pop_approx}, which provides
the existence of the point $\left(m_{1},m_{2}\right)\in\IR^{2}$ as
the limit of the moments
\[
\IE_{J,N_{1},N_{2}}\frac{S_{\lambda}S_{\nu}}{N_{\lambda}N_{\nu}}
\]
as $N_{1},N_{2}\rightarrow\infty$, the point $\left(m_{1},m_{2}\right)\in\IR^{2}$
can also be defined for the high temperature regime $I-J>0$. The
value $\left(m_{1},m_{2}\right)$ would be $0$ by the proposition.
Hence, by setting $\hat{m}_{N_{1},N_{2},n}\left(x\right)\coloneq0$
for samples $x$ where $T\left(x\right)$ is singular or $I-\boldsymbol{N}T\left(x\right)^{-1}\boldsymbol{N}>0$
does not hold, we obtain a valid estimator for $\left(m_{1},m_{2}\right)$
for both regimes.
\end{rem}

\section{\label{sec:Results}Theoretical Results}
\begin{notation}
Let $\left(Y_{n}\right)_{n\in\IN}$ be a sequence of random variables,
let $Y$ be a random variable, and let $P$ be a probability measure.
Then the notation
\[
Y_{n}\xrightarrow[n\rightarrow\infty]{\textup{p}}Y
\]
stands for the statement `$Y_{n}$ converges in probability to $Y$
as $n$ goes to infinity,' and
\[
Y_{n}\xrightarrow[n\rightarrow\infty]{\textup{d}}P
\]
means `$Y_{n}$ converges in distribution to $P$ as $n$ goes to
infinity.'

Let $A\in\IR^{d\times d},A>0$. We will write $\mathcal{N}\left(0,A\right)$
for the multivariate normal distribution on $\IR^{d}$ with mean 0
and covariance matrix $A$.
\end{notation}

The main theoretical result of this article concerns the asymptotic
(in terms of large samples and large voting populations, respectively)
behaviour of the estimators defined in Section \ref{sec:Estimation}.
The following theorem sums up these results.
\begin{thm}
\label{thm:limit}Let $N_{1},N_{2}\in\IN$, $J>0$, and $I-J>0$.
Then the following statements hold:
\begin{enumerate}
\item $\hat{J}_{N_{1},N_{2},n}\xrightarrow[n\rightarrow\infty]{\textup{p}}\tilde{J}_{N_{1},N_{2}}$.
\item $\tilde{J}_{N_{1},N_{2}}\xrightarrow[N_{1},N_{2}\rightarrow\infty]{}J$.
\item Let $\text{\underbar{\ensuremath{\hat{J}}}\ensuremath{\coloneq}\ensuremath{\left(\left(\hat{J}_{N_{1},N_{2},n}\right)_{1,1},\left(\hat{J}_{N_{1},N_{2},n}\right)_{2,2},\left(\hat{J}_{N_{1},N_{2},n}\right)_{1,2}\right)}}$
and $\underbar{\ensuremath{\tilde{J}}}\coloneq\left(\left(\tilde{J}_{N_{1},N_{2}}\right)_{1,1},\left(\tilde{J}_{N_{1},N_{2}}\right)_{2,2},\left(\tilde{J}_{N_{1},N_{2}}\right)_{1,2}\right)$.
Then $\sqrt{n}\left(\underbar{\ensuremath{\hat{J}}}-\underbar{\ensuremath{\tilde{J}}}\right)\xrightarrow[n\rightarrow\infty]{\textup{d}}\mathcal{N}\left(0,C\right)$
, where $C>0$ is the $3\times3$ covariance matrix given in (\ref{eq:lim_cov}).
\end{enumerate}
Now assume $I-J\geq0$ does not hold. Set $\underbar{\ensuremath{\hat{m}}}\ensuremath{\coloneq}\left(\left(\hat{m}_{N_{1},N_{2},n}\right)_{1,1},\left(\hat{m}_{N_{1},N_{2},n}\right)_{2,2},\left(\hat{m}_{N_{1},N_{2},n}\right)_{1,2}\right)$
and\\
$\ensuremath{\underbar{\ensuremath{\tilde{m}}}}\ensuremath{\coloneq}\left(\left(\tilde{m}_{N_{1},N_{2}}\right)_{1,1},\left(\tilde{m}_{N_{1},N_{2}}\right)_{2,2},\left(\tilde{m}_{N_{1},N_{2}}\right)_{1,2}\right)$.
\begin{enumerate}
\item $\hat{m}_{N_{1},N_{2},n}\xrightarrow[n\rightarrow\infty]{\textup{p}}\tilde{m}_{N_{1},N_{2}}$.
\item $\tilde{m}_{N_{1},N_{2}}\xrightarrow[N_{1},N_{2}\rightarrow\infty]{}\left(\begin{array}{cc}
m^{2}_{1} & m_{1}m_{2}\\
m_{1}m_{2} & m^{2}_{2}
\end{array}\right)$.
\item $\sqrt{n}\left(\underbar{\ensuremath{\hat{m}}}-\ensuremath{\underbar{\ensuremath{\tilde{m}}}}\right)\xrightarrow[n\rightarrow\infty]{\textup{d}}\mathcal{N}\left(0,D\right)$,
where $D$ is the covariance matrix of the random vector $\left(\left(\frac{S_{1}}{N_{1}}\right)^{2},\left(\frac{S_{2}}{N_{2}}\right)^{2},\frac{S_{1}S_{2}}{N_{1}N_{2}}\right)$.
\end{enumerate}
\end{thm}

See Section \ref{subsec:Proof-of-Theorem} for the proof of this theorem.

The previous theorem states that in the high temperature regime the
estimator $\hat{J}_{N_{1},N_{2},n}$ is approximately consistent (meaning
it is consistent in the large population limit) and is asymptotically
normal for large sample sizes. In the low temperature regime, the
estimator $\hat{m}_{N_{1},N_{2},n}$ for the limit of the expectations
$\IE_{J,N_{1},N_{2}}\frac{S_{\lambda}S_{\nu}}{N_{\lambda}N_{\nu}}$
is approximately consistent and asymptotically normal for large sample
sizes.

A secondary result about the relation between the concept of independent
groups and the correlation between the group voting margins $S_{1}$
and $S_{2}$ will be helpful for running hypothesis tests for the
independence of groups (see Section \ref{sec:Simulations}) and is
the subject of the following lemma.
\begin{lem}
\label{lem:indep_corr}The statements below are equivalent:
\end{lem}

\begin{enumerate}
\item $J_{1,2}=0$,
\item $S_{1}$ and $S_{2}$ are independent random variables,
\item $\textup{corr}\left(S_{1},S_{2}\right)=0$.
\end{enumerate}
See Section \ref{subsec:Proof-of-Lemma} for the proof of this lemma.

\section{\label{sec:Simulations}Simulations of the Estimation Procedure}

In addition to the theoretical results in the previous section, we
simulated the estimation procedure to identify the regime of the model
and estimate the coupling parameters for values in the high and the
low temperature regime. The main questions we wanted to answer were:
\begin{enumerate}
\item How likely is it that the estimator is realised in the wrong regime
(i.e.\! a low temperature result when the actual coupling matrix
is in high temperature regime or vice versa)?
\item How likely is a deviation of a certain magnitude from the true parameter
values for different sample sizes?
\item How likely is it that a dependency in the voting behaviour of the
two groups can be detected in terms of a hypothesis test to reject
the null hypothesis $J_{1,2}=0$?
\end{enumerate}
In order to answer these questions, we ran simulations to repeatedly
generate samples from the distribution in Definition \ref{def:CWM}
and calculated the estimators from Definitions \ref{def:high_temp_est} and \ref{def:low_temp_est}.
This yields a panorama of the approximate probabilities mentioned
above for the examples we considered.

\subsection{Identification of the Model's Regime and Parameter Estimation}

We ran simulations for group sizes $N_{1}=N_{2}=100$ and $N_{1}=N_{2}=1000$.
There were 100 essays for each combination of group and sample sizes.

\subsubsection{High Temperature Regime}

We picked a high temperature coupling matrix:
\begin{equation}
J=\left(\begin{array}{cc}
0.6 & 0.3\\
0.3 & 0.5
\end{array}\right).\label{eq:high_temp_example}
\end{equation}
For this $J$ and groups of size $N_{1}=N_{2}=100$, we generated
100 samples of size $n=20$ each. The regime was correctly identified
as high temperature from 73\% of the samples.

We employed the norm $\left\Vert \cdot\right\Vert :\IR^{2\times2}\rightarrow\IR,\left\Vert A\right\Vert \coloneq\max\left\{ \left|A_{i,j}\right|,i,j\in\IN_{2}\right\} $,
$A\in\IR^{2\times2}$, to measure the distance $\left\Vert \hat{J}_{N_{1},N_{2},n}\left(x\right)-J\right\Vert $
for each sample generated. The estimate $\hat{J}_{N_{1},N_{2},n}\left(x\right)$
lay within a tolerance of 0.1 of the true coupling matrix in (\ref{eq:high_temp_example})
20\% of the time.

Next, we generated 100 samples of size $n=100$ each. The regime was
correctly identified as high temperature from 96\% of the samples.
The estimate $\hat{J}_{N_{1},N_{2},n}\left(x\right)$ lay within a
tolerance of 0.1 of the true coupling matrix in (\ref{eq:high_temp_example})
64\% of the time.

Finally, we generated 100 samples of size $n=500$ each. The regime
was correctly identified as high temperature from all samples. The
estimate $\hat{J}_{N_{1},N_{2},n}\left(x\right)$ lay within a tolerance
of 0.1 of the true coupling matrix in (\ref{eq:high_temp_example})
96\% of the time.

See Table \ref{tab:high_temp_100} for a summary of these results.

\begin{table}
\begin{centering}
\renewcommand{\arraystretch}{1.5}%
\begin{tabular}{|c|c|c|c|c|}
\hline 
$n$ & Regime correct & Within tol 0.1 & Within tol 0.05 & Within tol 0.025\tabularnewline
\hline 
\hline 
20 & 73\% & 27\% & 4\% & 1\%\tabularnewline
\hline 
100 & 96\% & 64\% & 16\% & 2\%\tabularnewline
\hline 
500 & 100\% & 96\% & 70\% & 14\%\tabularnewline
\hline 
\end{tabular}
\par\end{centering}
\caption{Simulation results for high temperature coupling matrix (\ref{eq:high_temp_example})
with group sizes $N_{1}=N_{2}=100$ \label{tab:high_temp_100}}

\end{table}
We repeated the same simulation for groups of size $N_{1}=N_{2}=1000$.
While the frequencies of correctly identifying the regime as high
temperature and the estimate $\hat{J}_{N_{1},N_{2},n}\left(x\right)$
being within a tolerance of 0.1 of the true $J$ are similar to the
case of the smaller groups of size 100, the frequencies at which tighter
tolerances are achieved are higher for the larger groups of size 1000
and sample sizes of 100 and 500. See Table \ref{tab:high_temp_1000}
for all the results. 
\begin{table}
\begin{centering}
\renewcommand{\arraystretch}{1.5}{\footnotesize{}%
\begin{tabular}{|c|c|c|c|c|c|c|}
\hline 
{\footnotesize$n$} & {\footnotesize Regime correct} & {\footnotesize Within tol 0.1} & {\footnotesize Within tol 0.05} & {\footnotesize Within tol 0.025} & {\footnotesize Within tol 0.01} & {\footnotesize Within tol 0.005}\tabularnewline
\hline 
\hline 
{\footnotesize 20} & {\footnotesize 73\%} & {\footnotesize 24\%} & {\footnotesize 5\%} & {\footnotesize 1\%} & {\footnotesize 0\%} & {\footnotesize 0\%}\tabularnewline
\hline 
{\footnotesize 100} & {\footnotesize 91\%} & {\footnotesize 76\%} & {\footnotesize 35\%} & {\footnotesize 9\%} & {\footnotesize 0\%} & {\footnotesize 0\%}\tabularnewline
\hline 
{\footnotesize 500} & {\footnotesize 99\%} & {\footnotesize 99\%} & {\footnotesize 86\%} & {\footnotesize 51\%} & {\footnotesize 13\%} & {\footnotesize 3\%}\tabularnewline
\hline 
\end{tabular}}{\footnotesize\par}
\par\end{centering}
\caption{Simulation results for high temperature coupling matrix (\ref{eq:high_temp_example})
with group sizes $N_{1}=N_{2}=1000$ \label{tab:high_temp_1000}}
\end{table}

\subsubsection{Low Temperature Regime}

We picked a low temperature coupling matrix:
\begin{equation}
J=\left(\begin{array}{cc}
2 & 1\\
1 & 1.8
\end{array}\right).\label{eq:low_temp_example}
\end{equation}

By Proposition \ref{prop:large_pop_approx}, for the low temperature
regime, there is a point $m=\left(m_{1},m_{2}\right)\in\IR^{2}$ with
$m_{1}>0$ and $m_{2}\neq0$ such that
\[
\lim_{N_{1},N_{2}\rightarrow\infty}\IE\left(\frac{S_{\lambda}}{N_{\lambda}}\right)^{2}=m^{2}_{\lambda}.
\]
The point $m$ can be calculated numerically at a constant computational
cost independent of the group sizes. We can use the above limit to
estimate the concentration point $m$ from a sample of voting configurations
from the distribution in Definition \ref{def:CWM} with the coupling
matrix in (\ref{eq:low_temp_example}). Our estimator for $\left(m^{2}_{1},m^{2}_{2}\right)$
is $\left(\left(\underbar{\ensuremath{\hat{m}}}\right)^{2}_{1},\left(\underbar{\ensuremath{\hat{m}}}\right)^{2}_{2}\right)$
and the estimator for $\textup{sign }m_{2}$ is $\textup{sign \ensuremath{\left(\underbar{\ensuremath{\hat{m}}}\right)_{3}}}$.

For the coupling matrix $J$ in (\ref{eq:low_temp_example}) and groups
of size $N_{1}=N_{2}=100$, we generated 100 samples each of sizes
20, 100, and 500, respectively. The low temperature regime of the
model was correctly identified in all cases. Similarly, the estimator
$\left(\left(\underbar{\ensuremath{\hat{m}}}\right)^{2}_{1},\left(\underbar{\ensuremath{\hat{m}}}\right)^{2}_{2}\right)$
was realised within a tolerance of 0.1 of the true value $\left(m^{2}_{1},m^{2}_{2}\right)$
100\% of the time. See Table \ref{tab:low_temp_100}.

\begin{table}
\begin{centering}
\renewcommand{\arraystretch}{1.5}%
\begin{tabular}{|c|c|c|c|c|c|}
\hline 
$n$ & Regime correct & Within tol 0.1 & Within tol 0.05 & Within tol 0.025 & Within tol 0.01\tabularnewline
\hline 
\hline 
20 & 100\% & 100\% & 100\% & 100\% & 87\%\tabularnewline
\hline 
100 & 100\% & 100\% & 100\% & 100\% & 100\%\tabularnewline
\hline 
500 & 100\% & 100\% & 100\% & 100\% & 100\%\tabularnewline
\hline 
\end{tabular}\vspace{.5\baselineskip}
\\
\begin{tabular}{|c|c|c|c|c|c|}
\hline 
$n$ & Within tol 0.005 & Within tol 0.0025 & Within tol 0.001 & Within tol 0.0005 & Within tol 0.00025\tabularnewline
\hline 
\hline 
20 & 48\% & 17\% & 0\% & 0\% & 0\%\tabularnewline
\hline 
100 & 94\% & 60\% & 16\% & 5\% & 0\%\tabularnewline
\hline 
500 & 100\% & 97\% & 31\% & 0\% & 0\%\tabularnewline
\hline 
\end{tabular}
\par\end{centering}
\caption{Simulation results for low temperature coupling matrix (\ref{eq:low_temp_example})
with group sizes $N_{1}=N_{2}=100$\label{tab:low_temp_100}}
\end{table}
For the same coupling matrix in (\ref{eq:low_temp_example}) and group
sizes $N_{1}=N_{2}=1000$, we ran simulations for the same sample
sizes 20, 100, and 500. Similarly to the high temperature example
in the previous section, we found that for larger groups tighter tolerances
are achieved. However, for both group sizes, we noted an aberration
as the tightest tolerances are achieved at a slightly higher (but
still very low) frequency for samples of size 100 than for 500. We
interpret this as being a case where the higher variance of the smaller
of the two sample sizes occasionally works in our favour. The most
likely explanation is that since the actual expectations $\IE\left(\frac{S_{\lambda}}{N_{\lambda}}\right)^{2}$
are close to but not equal to their limit $m^{2}_{\lambda}$ for $N_{\lambda}=1000$,
the larger samples of size 500 allow us to better approximate the
true value $\IE\left(\frac{S_{\lambda}}{N_{\lambda}}\right)^{2}$
compared to samples of moderate size 100. 
\begin{table}
\begin{centering}
\renewcommand{\arraystretch}{1.5}{\footnotesize{}%
\begin{tabular}{|c|c|c|c|c|c|c|}
\hline 
{\footnotesize$n$} & {\footnotesize Regime correct} & {\footnotesize Within tol 0.1} & {\footnotesize Within tol 0.05} & {\footnotesize Within tol 0.025} & {\footnotesize Within tol 0.01} & {\footnotesize Within tol 0.005}\tabularnewline
\hline 
\hline 
{\footnotesize 20} & {\footnotesize 100\%} & {\footnotesize 100\%} & {\footnotesize 100\%} & {\footnotesize 100\%} & {\footnotesize 100\%} & {\footnotesize 100\%}\tabularnewline
\hline 
{\footnotesize 100} & {\footnotesize 100\%} & {\footnotesize 100\%} & {\footnotesize 100\%} & {\footnotesize 100\%} & {\footnotesize 100\%} & {\footnotesize 100\%}\tabularnewline
\hline 
{\footnotesize 500} & {\footnotesize 100\%} & {\footnotesize 100\%} & {\footnotesize 100\%} & {\footnotesize 100\%} & {\footnotesize 100\%} & {\footnotesize 100\%}\tabularnewline
\hline 
\end{tabular}}\vspace{.5\baselineskip}
\\
{\footnotesize{}%
\begin{tabular}{|c|c|c|c|c|c|}
\hline 
{\footnotesize$n$} & {\footnotesize Within tol 0.0025} & {\footnotesize Within tol 0.001} & {\footnotesize Within tol 0.0005} & {\footnotesize Within tol 0.00025} & {\footnotesize Within tol 0.0001}\tabularnewline
\hline 
\hline 
{\footnotesize 20} & {\footnotesize 81\%} & {\footnotesize 24\%} & {\footnotesize 7\%} & {\footnotesize 2\%} & {\footnotesize 0\%}\tabularnewline
\hline 
{\footnotesize 100} & {\footnotesize 100\%} & {\footnotesize 65\%} & {\footnotesize 31\%} & {\footnotesize 11\%} & {\footnotesize 2\%}\tabularnewline
\hline 
{\footnotesize 500} & {\footnotesize 100\%} & {\footnotesize 97\%} & {\footnotesize 41\%} & {\footnotesize 6\%} & {\footnotesize 1\%}\tabularnewline
\hline 
\end{tabular}}{\footnotesize\par}
\par\end{centering}
\caption{Simulation results for low temperature coupling matrix (\ref{eq:low_temp_example})
with group sizes $N_{1}=N_{2}=1000$\label{tab:low_temp-1000}}
\end{table}

\subsection{Hypothesis Tests for Independent Groups}

We want to do hypothesis tests to discard the possibility that the
two groups are independent, i.e.\! $J_{1,2}=0$. By Lemma \ref{lem:indep_corr},
$J_{1,2}=0$ if and only if the $\textup{corr}\left(S_{1},S_{2}\right)=0$.
Thus, we can use the realisations of the random variables $S_{1}$
and $S_{2}$ to test for independence.

Hence, the hypothesis test
\[
h_{0}:J_{1,1}=0,\quad h_{1}:J_{1,1}\neq0
\]
can equivalently be formulated as
\begin{equation}
h_{0}:\textup{corr}\left(S_{1},S_{2}\right)=0,\quad h_{1}:\textup{corr}\left(S_{1},S_{2}\right)\neq0.\label{eq:hypo_test}
\end{equation}
Set for a given sample $x=\left(x^{(1)},\ldots,x^{(n)}\right)\in\Omega^{n}_{N_{1}+N_{2}}$
\[
\bar{x}_{\lambda}\coloneq\sum^{N_{\lambda}}_{i=1}x^{(t)}_{\lambda,i},\quad\lambda\in\IN_{2}.
\]

\begin{defn}
Pearson's correlation coefficient for $x=\left(x^{(1)},\ldots,x^{(n)}\right)\in\Omega^{n}_{N_{1}+N_{2}}$
is defined by the formula
\[
r\left(x\right)\coloneq\frac{\sum^{n}_{t=1}\left(\sum^{N_{1}}_{i=1}x^{(t)}_{1,i}-\bar{x}_{1}\right)\left(\sum^{N_{2}}_{j=1}x^{(t)}_{2,j}-\bar{x}_{2}\right)}{\sqrt{\sum^{n}_{t=1}\left(\sum^{N_{1}}_{i=1}x^{(t)}_{1,i}-\bar{x}_{1}\right)^{2}}\sqrt{\sum^{n}_{t=1}\left(\sum^{N_{2}}_{j=1}x^{(t)}_{2,j}-\bar{x}_{2}\right)^{2}}}.
\]
\end{defn}

The function $r$ is a statistic, whereas the parameter
\[
\rho\coloneq\textup{corr}\left(S_{1},S_{2}\right)
\]
depends on the distribution $\mathbb{P}_{J,N_{1},N_{2}}$ from Definition
\ref{def:CWM}. By the law of large numbers,
\[
r\xrightarrow[n\rightarrow\infty]{\textup{p}}\rho
\]
holds.

In order to do the hypothesis test (\ref{eq:hypo_test}), we apply
the Fisher transformation. Let `artanh' below stand for the inverse
function of the hyperbolic tangent $\tanh$.
\begin{defn}
For each sample $x=\left(x^{(1)},\ldots,x^{(n)}\right)\in\Omega^{n}_{N_{1}+N_{2}}$,
set
\[
F\left(x\right)\coloneq\begin{cases}
\artanh r\left(x\right) & \textup{if }\left|r\left(x\right)\right|<1,\\
\infty & \textup{if }r\left(x\right)=1,\\
-\infty & \textup{if }r\left(x\right)=-1.
\end{cases}
\]

It is a well-known fact that $F$ is approximately normally distributed
according to $\mathcal{N}\left(\artanh\rho,\frac{1}{n-3}\right)$
as $n\rightarrow\infty$ (see, e.g., \cite[ Chapter 14]{Kendall1945}).

In all simulations below, a level of significance of $0.05$ was used
to decide whether to reject the null hypothesis. We chose group sizes
$N_{1}=N_{2}=100$ for the simulations in this section.
\end{defn}

\subsubsection{High Temperature Regime}

We simulated the hypothesis test (\ref{eq:hypo_test}) for samples
from a high temperature coupling matrix
\begin{equation}
J=\left(\begin{array}{cc}
0.6 & J_{1,2}\\
J_{1,2} & 0.5
\end{array}\right)\label{eq:high_temp_small_J11}
\end{equation}
with three different values $J_{1,2}\in\left\{ 0.05,0.1,0.2\right\} $.
We once again simulated three different sample sizes $n\in\left\{ 20,100,500\right\} $,
doing 100 essays for each sample size.

\begin{table}
\begin{centering}
\renewcommand{\arraystretch}{1.5}%
\begin{tabular}{|c|c|c|c|}
\hline 
$n$ & $J_{1,2}=0.2$ & $J_{1,2}=0.1$ & $J_{1,2}=0.05$\tabularnewline
\hline 
\hline 
20 & 48\% & 14\% & 2\%\tabularnewline
\hline 
100 & 100\% & 62\% & 17\%\tabularnewline
\hline 
500 & 100\% & 100\% & 71\%\tabularnewline
\hline 
\end{tabular}
\par\end{centering}
\caption{Rejection frequency of $h_{0}$ for the high temperature coupling
matrix (\ref{eq:high_temp_small_J11})\label{tab:high_temp_J11}}
\end{table}
For a sample size of $n=20$, a value of even $J_{1,2}=0.2$ cannot
be reliably detected, as the null hypothesis $\corr\left(S_{1},S_{2}\right)=0$
could only be rejected in 48\% of the essays. Smaller values of $J_{1,2}$
are even more difficult to detect. Samples of size $n=100$ very reliably
allow for the detection of dependence between the groups if $J_{1,2}=0.2$.
For smaller values of $J_{1,2}$ reliability decreases sharply. Large
samples of size $n=500$ allow for reliable detection if $J_{1,2}=0.1$
and even $J_{1,2}=0.05$ allowed for the rejection of the null hypothesis
in 71\% of all essays. All these results can be found in Table \ref{tab:high_temp_J11}.

\subsubsection{Low Temperature Regime}

Next, we picked a low temperature coupling matrix
\begin{equation}
J=\left(\begin{array}{cc}
2 & J_{1,2}\\
J_{1,2} & 1.8
\end{array}\right)\label{eq:low_temp_small_J11}
\end{equation}
with the same values of $J_{1,2}\in\left\{ 0.05,0.1,0.2\right\} $
as in the high temperature example above. 100 samples each were generated
of sizes $n\in\left\{ 20,100,500\right\} $.
\begin{table}
\begin{centering}
\renewcommand{\arraystretch}{1.5}%
\begin{tabular}{|c|c|c|c|}
\hline 
$n$ & $J_{1,2}=0.2$ & $J_{1,2}=0.1$ & $J_{1,2}=0.05$\tabularnewline
\hline 
\hline 
20 & 86\% & 29\% & 15\%\tabularnewline
\hline 
100 & 100\% & 94\% & 40\%\tabularnewline
\hline 
500 & 100\% & 100\% & 98\%\tabularnewline
\hline 
\end{tabular}
\par\end{centering}
\caption{Rejection frequency of $h_{0}$ for the low temperature coupling matrix
(\ref{eq:low_temp_small_J11})\label{tab:low_temp_J11}}
\end{table}
For the low temperature coupling matrix (\ref{eq:low_temp_small_J11}),
even small samples of size $n=20$, allowed us to reject the null
hypothesis of independence when $J_{1,2}=0.2$ in 86\% of all essays.
For smaller values of $J_{1,2}$ the rejection rates decline sharply,
but are higher than for the high temperature coupling matrix (\ref{eq:high_temp_small_J11}).
Samples of size $n=100$ reliably allow for the detection of dependence
between the groups if $J_{1,2}=0.1$, which is a point of distinction
compared to the high temperature example. Finally, large samples of
size $n=500$ allows for reliable detection even if $J_{1,2}=0.05$,
in which case the null hypothesis was rejected in 98\% of all essays.
These results can be found in Table \ref{tab:low_temp_J11}.

We conclude that even though the set of values for $J_{1,2}\in\left\{ 0.05,0.1,0.2\right\} $
is the same in both examples (\ref{eq:high_temp_small_J11}) and (\ref{eq:low_temp_small_J11}),
the more ordered state of the model in the low temperature regime
allows us to reject the hypothesis of independence of groups at a
higher rate than under the more chaotic conditions of the high temperature
regime. This illustrates the more complicated structure of correlations
between votes from different groups when $J_{1,2}\neq0$ mentioned
in Remark \ref{rem:indep_groups}.

As a concluding remark concerning all simulation results in this section,
we would like to draw attention to the fact that the group sizes of
$N_{1}=N_{2}=100$ and $N_{1}=N_{2}=1000$ considered in the simulations
are rather small when compared to real world applications. It stands
to reason that the results indicate that even for fairly small populations,
the asymptotic approximations employed by the estimators yield good
results. If the groups are of larger sizes, such as in the case of
groups representing states or countries, the quality of the estimators
improves further.

\section{\label{sec:Proofs}Proofs of the Results from Section \ref{sec:Results}}

\subsection{\label{subsec:Proof-of-Theorem}Proof of Theorem \ref{thm:limit}}

\subsubsection{Statement 1}

Recall that $x=\left(x^{(1)},\ldots,x^{(n)}\right)\in\Omega^{n}_{N_{1}+N_{2}}$
is a sample of voting configurations we observe, the realisation of
the random variable $X=\left(X^{(1)},\ldots,X^{(n)}\right)$. The
weak law of large numbers
\begin{equation}
\left(T\left(X\right)\right)_{1,1}=\frac{1}{n}\sum^{n}_{t=1}\left(\sum^{N_{1}}_{i=1}X^{(t)}_{1,i}\right)^{2}\xrightarrow[n\rightarrow\infty]{\textup{p}}\mathbb{E}_{J,N_{1},N_{2}}S^{2}_{1}\label{eq:WLLN}
\end{equation}
holds because $\left(\sum^{N_{1}}_{i=1}X^{(t)}_{1,i}\right)^{2}$
is a bounded random variable, and thus the second moment exists. The
same argument yields the other two statements akin to (\ref{eq:WLLN}),
and therefore
\[
T\xrightarrow[n\rightarrow\infty]{\textup{p}}\left(\begin{array}{cc}
\IE_{J,N_{1},N_{2}}S^{2}_{1} & \IE_{J,N_{1},N_{2}}S_{1}S_{2}\\
\IE_{J,N_{1},N_{2}}S_{1}S_{2} & \IE_{J,N_{1},N_{2}}S^{2}_{2}
\end{array}\right).
\]
For the regime where $I-J\geq0$ does not hold, this shows the statement
$\hat{m}_{N_{1},N_{2},n}\xrightarrow[n\rightarrow\infty]{\textup{p}}\tilde{m}_{N_{1},N_{2}}$. 

Suppose $I-J>0$ holds. By Definition \ref{def:J_tilde}, we have
\[
\left(\begin{array}{cc}
\IE_{J,N_{1},N_{2}}\frac{S^{2}_{1}}{N_{1}} & \IE_{J,N_{1},N_{2}}\frac{S_{1}S_{2}}{\sqrt{N_{1}N_{2}}}\\
\IE_{J,N_{1},N_{2}}\frac{S_{1}S_{2}}{\sqrt{N_{1}N_{2}}} & \IE_{J,N_{1},N_{2}}\frac{S^{2}_{2}}{N_{2}}
\end{array}\right)=\left(I-\tilde{J}_{N_{1},N_{2}}\right)^{-1}.
\]
Recall the matrix $\boldsymbol{N}$ from Definition (\ref{def:high_temp_est}).
The last two displays combine to
\[
T\xrightarrow[n\rightarrow\infty]{\textup{p}}\boldsymbol{N}\left(I-\tilde{J}_{N_{1},N_{2}}\right)^{-1}\boldsymbol{N}.
\]
Definition \ref{def:high_temp_est} states that the estimator $\hat{J}_{N_{1},N_{2},n}$
is given by
\[
\hat{J}_{N_{1},N_{2},n}=I-\boldsymbol{N}T^{-1}\boldsymbol{N}
\]
if $T$ is invertible and $I-\boldsymbol{N}T^{-1}\boldsymbol{N}>0$
holds.

The mapping $A\mapsto I-\boldsymbol{N}A^{-1}\boldsymbol{N}$ for invertible
$2\times2$ matrices $A$ is continuous. Hence, Theorem \ref{thm:cont_mapping}
yields
\[
\hat{J}_{N_{1},N_{2},n}\xrightarrow[n\rightarrow\infty]{\textup{p}}\tilde{J}_{N_{1},N_{2}}.
\]

\subsubsection{Statement 2}

Next, we prove the statements $\tilde{J}_{N_{1},N_{2}}\xrightarrow[N_{1},N_{2}\rightarrow\infty]{}J$
and $\tilde{m}_{N_{1},N_{2}}\xrightarrow[N_{1},N_{2}\rightarrow\infty]{}\left(\begin{array}{cc}
m^{2}_{1} & m_{1}m_{2}\\
m_{1}m_{2} & m^{2}_{2}
\end{array}\right)$.

Proposition \ref{prop:large_pop_approx} says that if $I-J>0$ holds,
then we have, for all $\lambda,\nu\in\IN_{2}$,
\[
\IE_{J,N_{1},N_{2}}\frac{S_{\lambda}S_{\nu}}{\sqrt{N_{\lambda}N_{\nu}}}\xrightarrow[N_{\lambda},N_{\nu}\rightarrow\infty]{}\left(\left(I-J\right)^{-1}\right)_{\lambda,\nu}.
\]
By Definition \ref{def:J_tilde}, for all $N_{1},N_{2}\in\IN$,
\[
\left(I-\tilde{J}_{N_{1},N_{2}}\right)^{-1}=\left(\begin{array}{cc}
\IE_{J,N_{1},N_{2}}\frac{S^{2}_{1}}{N_{1}} & \IE_{J,N_{1},N_{2}}\frac{S_{1}S_{2}}{\sqrt{N_{1}N_{2}}}\\
\IE_{J,N_{1},N_{2}}\frac{S_{1}S_{2}}{\sqrt{N_{1}N_{2}}} & \IE_{J,N_{1},N_{2}}\frac{S^{2}_{2}}{N_{2}}
\end{array}\right)\xrightarrow[N_{\lambda},N_{\nu}\rightarrow\infty]{}\left(I-J\right)^{-1},
\]
which, due to the continuity of the mapping $A\mapsto\left(I-A\right)^{-1}$,
is equivalent to $\tilde{J}_{N_{1},N_{2}}\xrightarrow[N_{\lambda},N_{\nu}\rightarrow\infty]{}J$.

Next, assume $I-J\geq0$ does not hold. Definition \ref{def:J_tilde}
states that, for all $N_{1},N_{2}\in\IN$,
\[
\tilde{m}_{N_{1},N_{2}}=\left(\begin{array}{cc}
\IE_{J,N_{1},N_{2}}\left(\frac{S_{1}}{N_{1}}\right)^{2} & \IE_{J,N_{1},N_{2}}\frac{S_{1}S_{2}}{N_{1}N_{2}}\\
\IE_{J,N_{1},N_{2}}\frac{S_{1}S_{2}}{N_{1}N_{2}} & \IE_{J,N_{1},N_{2}}\left(\frac{S_{2}}{N_{2}}\right)^{2}
\end{array}\right),
\]
and by Proposition \ref{prop:large_pop_approx}, we have
\[
\IE_{J,N_{1},N_{2}}\frac{S_{\lambda}S_{\nu}}{N_{\lambda}N_{\nu}}\xrightarrow[N_{\lambda},N_{\nu}\rightarrow\infty]{}m_{\lambda}m_{\nu}
\]
for all $\lambda,\nu\in\IN_{2}$. Combining the last two displays,
we arrive at
\[
\tilde{m}_{N_{1},N_{2}}\xrightarrow[N_{\lambda},N_{\nu}\rightarrow\infty]{}\left(\begin{array}{cc}
m^{2}_{1} & m_{1}m_{2}\\
m_{1}m_{2} & m^{2}_{2}
\end{array}\right).
\]

\subsubsection{Statement 3}

We first show the result for $I-J>0$. We recall Definition \ref{def:T_stat}
of the statistic
\[
T\left(x\right)\coloneq\frac{1}{n}\sum^{n}_{t=1}\left(\begin{array}{cc}
\left(\sum^{N_{1}}_{i=1}x^{(t)}_{1,i}\right)^{2} & \sum^{N_{1}}_{i=1}x^{(t)}_{1,i}\sum^{N_{2}}_{j=1}x^{(t)}_{2,j}\\
\sum^{N_{1}}_{i=1}x^{(t)}_{1,i}\sum^{N_{2}}_{j=1}x^{(t)}_{2,j} & \left(\sum^{N_{2}}_{j=1}x^{(t)}_{2,j}\right)^{2}
\end{array}\right).
\]

For the random vectors $\left(T_{1,1},T_{2,2},T_{1,2}\right)$, we
have the weak law of large numbers
\[
\left(T_{1,1},T_{2,2},T_{1,2}\right)\xrightarrow[n\rightarrow\infty]{\textup{p}}\left(\IE_{J,N_{1},N_{2}}S^{2}_{1},\IE_{J,N_{1},N_{2}}S^{2}_{2},\IE_{J,N_{1},N_{2}}S_{1}S_{2}\right).
\]

In addition, we also note that $\left(T_{1,1},T_{2,2},T_{1,2}\right)$
is the sum of i.i.d.\! random variables, each summand having expectation
\[
\mu\coloneq\left(\IE_{J,N_{1},N_{2}}S^{2}_{1},\IE_{J,N_{1},N_{2}}S^{2}_{2},\IE_{J,N_{1},N_{2}}S_{1}S_{2}\right)
\]
and covariance matrix
\[
D\coloneq
\global\long\def\arraystretch{1.3}%
\left(\begin{array}{ccc}
\IV\left(\frac{S_{1}}{N_{1}}\right)^{2} & \textup{Cov}\left(\left(\frac{S_{1}}{N_{1}}\right)^{2},\left(\frac{S_{2}}{N_{2}}\right)^{2}\right) & \textup{Cov}\left(\left(\frac{S_{1}}{N_{1}}\right)^{2},\frac{S_{1}S_{2}}{N_{1}N_{2}}\right)\\
\textup{Cov}\left(\left(\frac{S_{1}}{N_{1}}\right)^{2},\left(\frac{S_{2}}{N_{2}}\right)^{2}\right) & \IV\left(\frac{S_{2}}{N_{2}}\right)^{2} & \textup{Cov}\left(\left(\frac{S_{2}}{N_{2}}\right)^{2},\frac{S_{1}S_{2}}{N_{1}N_{2}}\right)\\
\textup{Cov}\left(\left(\frac{S_{1}}{N_{1}}\right)^{2},\frac{S_{1}S_{2}}{N_{1}N_{2}}\right) & \textup{Cov}\left(\left(\frac{S_{2}}{N_{2}}\right)^{2},\frac{S_{1}S_{2}}{N_{1}N_{2}}\right) & \IV\frac{S_{1}S_{2}}{N_{1}N_{2}}
\end{array}\right).
\]
The central limit theorem thus yields
\begin{equation}
\sqrt{n}\left(\left(\frac{T_{1,1}}{N^{2}_{1}},\frac{T_{2,2}}{N^{2}_{2}},\frac{T_{1,2}}{N_{1}N_{2}}\right)-\mu\right)\xrightarrow[n\rightarrow\infty]{\textup{d}}\mathcal{N}\left(0,D\right).\label{eq:CLT_T}
\end{equation}
for any regime. By Definitions \ref{def:T_stat} and \ref{def:low_temp_est},
this proves the statement $\sqrt{n}\left(\underbar{\ensuremath{\hat{m}}}-\ensuremath{\underbar{\ensuremath{\tilde{m}}}}\right)\xrightarrow[n\rightarrow\infty]{\textup{p}}\mathcal{N}\left(0,D\right)$
for the low temperature regime where $I-J\ngeq0$.

Now let $I-J>0$. Set $\mathbf{P}\coloneq\left\{ \left(s,t,u\right)\in\IR^{3}\,\left|\,\left(\begin{array}{cc}
s & u\\
u & t
\end{array}\right)>0\right.\right\} $. The functions $f,g,h:\mathbf{P}\rightarrow\IR$, defined for all
$\left(s,t,u\right)\in\mathbf{P}$ by $f\left(s,t,u\right)\coloneq\left(I-\boldsymbol{N}\left(\begin{array}{cc}
s & u\\
u & t
\end{array}\right)^{-1}\boldsymbol{N}\right)_{1,1}$, $g\left(s,t,u\right)\coloneq\left(I-\boldsymbol{N}\left(\begin{array}{cc}
s & u\\
u & t
\end{array}\right)^{-1}\boldsymbol{N}\right)_{2,2},$ and $h\left(s,t,u\right)\coloneq\left(I-\boldsymbol{N}\left(\begin{array}{cc}
s & u\\
u & t
\end{array}\right)^{-1}\boldsymbol{N}\right)_{1,2}$ are continuously differentiable on the open set
\[
\Gamma\coloneq\left\{ \left(s,t,u\right)\in\mathbf{P}\,\left|\,I-\boldsymbol{N}\left(\begin{array}{cc}
s & u\\
u & t
\end{array}\right)^{-1}\boldsymbol{N}>0\right.\right\} ,
\]
and we have
\begin{align*}
f\left(T_{1,1}\left(x\right),T_{2,2}\left(x\right),T_{1,2}\left(x\right)\right) & =\left(\hat{J}_{N_{1},N_{2},n}\left(x\right)\right)_{1,1}\\
g\left(T_{1,1}\left(x\right),T_{2,2}\left(x\right),T_{1,2}\left(x\right)\right) & =\left(\hat{J}_{N_{1},N_{2},n}\left(x\right)\right)_{2,2}\\
h\left(T_{1,1}\left(x\right),T_{2,2}\left(x\right),T_{1,2}\left(x\right)\right) & =\left(\hat{J}_{N_{1},N_{2},n}\left(x\right)\right)_{1,2},
\end{align*}
which is equivalent to
\[
\underbar{\ensuremath{\hat{J}}}\left(x\right)=\left(f\left(T_{1,1}\left(x\right),T_{2,2}\left(x\right),T_{1,2}\left(x\right)\right),g\left(T_{1,1}\left(x\right),T_{2,2}\left(x\right),T_{1,2}\left(x\right)\right),h\left(T_{1,1}\left(x\right),T_{2,2}\left(x\right),T_{1,2}\left(x\right)\right)\right).
\]
We note that the functions $f,g,h$ are continuously differentiable
and the set $\Gamma$ is an open subset of $\IR^{3}$. Thus, we can
apply Theorem \ref{thm:delta_method}. Let
\begin{align*}
\mu & \coloneq\underbar{\ensuremath{\tilde{J}}},\quad\Upsilon\coloneq D,\textup{ and}\\
\Delta\left(\underbar{\ensuremath{\tilde{J}}}\right) & \coloneq\left(\begin{array}{c}
\nabla f\left(\underbar{\ensuremath{\tilde{J}}}\right)^{T}\\
\nabla g\left(\underbar{\ensuremath{\tilde{J}}}\right)^{T}\\
\nabla h\left(\underbar{\ensuremath{\tilde{J}}}\right)^{T}
\end{array}\right),
\end{align*}
where
\[
\nabla f\coloneq\left(\begin{array}{c}
\frac{\partial f}{\partial s}\\
\frac{\partial f}{\partial t}\\
\frac{\partial f}{\partial u}
\end{array}\right)
\]
is the gradient of $f$, and $\nabla g$, $\nabla h$ are the gradients
of $g,h$, respectively.

Since (\ref{eq:CLT_T}) holds, we arrive at the conclusion that
\[
\sqrt{n}\left(\underbar{\ensuremath{\hat{J}}}-\underbar{\ensuremath{\tilde{J}}}\right)\xrightarrow[n\rightarrow\infty]{\textup{p}}\mathcal{N}\left(0,C\right)
\]
holds by Theorem \ref{thm:delta_method}. The limiting covariance
matrix $C$ is given by
\begin{align}
C & \coloneq\Delta\left(\mu\right)D\,\Delta\left(\mu\right)^{T}.\label{eq:lim_cov}
\end{align}

\subsection{\label{subsec:Proof-of-Lemma}Proof of Lemma \ref{lem:indep_corr}}

By Definition \ref{def:CWM},
\begin{align*}
\mathbb{P}_{J,N_{1},N_{2}}\left(X_{1,1}=x_{1,1},\ldots,X_{2,N_{2}}=x_{2,N_{2}}\right) & =Z^{-1}e^{-\mathbb{H}\left(x_{11},\ldots,x_{2N_{2}}\right)}\\
 & =Z^{-1}\exp\left(\frac{1}{2}\sum^{2}_{\lambda,\mu=1}\frac{J_{\lambda,\mu}}{\sqrt{N_{\lambda}N_{\mu}}}\sum^{N_{\lambda}}_{i=1}\sum^{N_{\mu}}_{j=1}x_{\lambda,i}x_{\mu,j}\right).
\end{align*}
Since the random variables $S_{1},S_{2}$ take the values $S_{1}\left(x\right)=\sum^{N_{1}}_{i=1}x_{1,i}$
and $S_{2}\left(x\right)=\sum^{N_{2}}_{j=1}x_{2,j}$ for each voting
configuration $x\in\Omega_{N_{1}+N_{2}}$, we can express the probability
in the previous display as
\[
\mathbb{P}_{J,N_{1},N_{2}}\left(X_{1,1}=x_{1,1},\ldots,X_{2,N_{2}}=x_{2,N_{2}}\right)=Z^{-1}\exp\left(\frac{1}{2}\sum^{2}_{\lambda,\mu=1}\frac{J_{\lambda,\mu}}{\sqrt{N_{\lambda}N_{\mu}}}\,S_{\lambda}\left(x\right)S_{\mu}\left(x\right)\right).
\]
The exponential above can be expressed as the product
\[
\exp\left(\frac{1}{2}\sum^{2}_{\lambda,\mu=1}\frac{J_{\lambda,\mu}}{\sqrt{N_{\lambda}N_{\mu}}}\,S_{\lambda}\left(x\right)S_{\mu}\left(x\right)\right)=\exp\left(\frac{J_{1,1}}{2N_{1}}\,S_{1}\left(x\right)^{2}\right)\exp\left(\frac{J_{2,2}}{2N_{2}}\,S_{2}\left(x\right)^{2}\right)
\]
for any voting configuration if and only if $J_{1,2}=0$.

Now write $\left(y_{1},\ldots,y_{N_{1}}\right)\coloneq\left(x_{1,1},\ldots,x_{1,N_{1}}\right)$
and $\left(z_{1},\ldots,z_{N_{2}}\right)\coloneq\left(x_{2,1},\ldots,x_{2,N_{2}}\right)$.
Since $S_{\lambda}$ only depends on the values $\left(x_{\lambda,1},\ldots,x_{\lambda,N_{\lambda}}\right)$,
the right hand side above can be written as
\[
\exp\left(\frac{J_{1,1}}{2N_{1}}\,S_{1}\left(x\right)^{2}\right)\exp\left(\frac{J_{2,2}}{2N_{2}}\,S_{2}\left(x\right)^{2}\right)=\exp\left(\frac{J_{1,1}}{2N_{1}}\,S_{1}\left(y\right)^{2}\right)\exp\left(\frac{J_{2,2}}{2N_{2}}\,S_{2}\left(z\right)^{2}\right).
\]
Thus, the random variables $S_{1}$ and $S_{2}$ are independent if
and only if $J_{1,2}=0$.

As uncorrelatedness follows from independence, it only remains to
prove that $\corr\left(S_{1},S_{2}\right)=0$ implies the independence
of $S_{1}$ and $S_{2}$. We show that if $J_{1,2}\neq0$, then $\corr\left(S_{1},S_{2}\right)\neq0$.
For this purpose, we first assume $J_{1,2}>0$.

We partition the set $\Omega_{N_{1}+N_{2}}$:
\begin{align*}
\Lambda_{+} & \coloneq\left\{ x\in\Omega_{N_{1}+N_{2}}\,|\,S_{1}\left(x\right)S_{2}\left(x\right)>0\right\} ,\\
\Lambda_{-} & \coloneq\left\{ x\in\Omega_{N_{1}+N_{2}}\,|\,S_{1}\left(x\right)S_{2}\left(x\right)<0\right\} ,\\
\Lambda_{0} & \coloneq\left\{ x\in\Omega_{N_{1}+N_{2}}\,|\,S_{1}\left(x\right)S_{2}\left(x\right)=0\right\} .
\end{align*}
Clearly, $\Lambda_{+},\Lambda_{-},\Lambda_{0}$ are disjoint and $\Lambda_{+}\cup\Lambda_{-}\cup\Lambda_{0}=\Omega_{N_{1}+N_{2}}$
holds. The observation that for all $x\in\Lambda_{+}$, the configuration
$\left(x_{1,1},\ldots,x_{1,N_{1}},-x_{2,1},\ldots,-x_{2,N_{2}}\right)$
belongs to $\Lambda_{-}$ allows us to define the function $f:\Lambda_{+}\rightarrow\Lambda_{-}$
by setting for all $x\in\Omega_{N_{1}},y\in\Omega_{N_{2}}$ such that
$\left(x,y\right)\in\Lambda_{+}$,
\begin{equation}
f\left(x,y\right)\coloneq\left(x,-y\right).\label{eq:bijection}
\end{equation}
A simple argument shows that $f$ is bijective.

We calculate
\begin{align*}
\corr\left(S_{1},S_{2}\right) & =\IE\,S_{1}S_{2}-\IE\,S_{1}\,\IE\,S_{2}=\IE\,S_{1}S_{2}\\
 & =Z^{-1}\sum_{x\in\Omega_{N_{1}+N_{2}}}S_{1}\left(x\right)S_{2}\left(x\right)\exp\left(\frac{1}{2}\sum^{2}_{\lambda,\mu=1}\frac{J_{\lambda,\mu}}{\sqrt{N_{\lambda}N_{\mu}}}\,S_{\lambda}\left(x\right)S_{\mu}\left(x\right)\right).
\end{align*}
Next, we write
\begin{align*}
 & \quad\sum_{x\in\Omega_{N_{1}+N_{2}}}S_{1}\left(x\right)S_{2}\left(x\right)\exp\left(\frac{1}{2}\sum^{2}_{\lambda,\mu=1}\frac{J_{\lambda,\mu}}{\sqrt{N_{\lambda}N_{\mu}}}\,S_{\lambda}\left(x\right)S_{\mu}\left(x\right)\right)\\
 & =\sum_{x\in\Lambda_{+}}S_{1}\left(x\right)S_{2}\left(x\right)\exp\left(\frac{1}{2}\sum^{2}_{\lambda,\mu=1}\frac{J_{\lambda,\mu}}{\sqrt{N_{\lambda}N_{\mu}}}\,S_{\lambda}\left(x\right)S_{\mu}\left(x\right)\right)\\
 & \quad+\sum_{x\in\Lambda_{-}}S_{1}\left(x\right)S_{2}\left(x\right)\exp\left(\frac{1}{2}\sum^{2}_{\lambda,\mu=1}\frac{J_{\lambda,\mu}}{\sqrt{N_{\lambda}N_{\mu}}}\,S_{\lambda}\left(x\right)S_{\mu}\left(x\right)\right)\\
 & \quad+\sum_{x\in\Lambda_{0}}S_{1}\left(x\right)S_{2}\left(x\right)\exp\left(\frac{1}{2}\sum^{2}_{\lambda,\mu=1}\frac{J_{\lambda,\mu}}{\sqrt{N_{\lambda}N_{\mu}}}\,S_{\lambda}\left(x\right)S_{\mu}\left(x\right)\right)\\
 & =\sum_{x\in\Lambda_{+}}S_{1}\left(x\right)S_{2}\left(x\right)\exp\left(\frac{1}{2}\sum^{2}_{\lambda,\mu=1}\frac{J_{\lambda,\mu}}{\sqrt{N_{\lambda}N_{\mu}}}\,S_{\lambda}\left(x\right)S_{\mu}\left(x\right)\right)\\
 & \quad+\sum_{x\in\Lambda_{-}}S_{1}\left(x\right)S_{2}\left(x\right)\exp\left(\frac{1}{2}\sum^{2}_{\lambda,\mu=1}\frac{J_{\lambda,\mu}}{\sqrt{N_{\lambda}N_{\mu}}}\,S_{\lambda}\left(x\right)S_{\mu}\left(x\right)\right).
\end{align*}
Taking into account that $f$ in (\ref{eq:bijection}) is bijective
and noting that $S_{1}\left(x\right)=S_{1}\left(f\left(x\right)\right)$
and $S_{2}\left(x\right)=-S_{2}\left(f\left(x\right)\right)$, we
can express the above sum as
\begin{align}
 & \sum_{x\in\Lambda_{+}}S_{1}\left(x\right)S_{2}\left(x\right)\exp\left(\frac{1}{2}\sum^{2}_{\lambda,\mu=1}\frac{J_{\lambda,\mu}}{\sqrt{N_{\lambda}N_{\mu}}}\,S_{\lambda}\left(x\right)S_{\mu}\left(x\right)\right)\nonumber \\
 & \quad+\sum_{x\in\Lambda_{+}}S_{1}\left(f\left(x\right)\right)S_{2}\left(f\left(x\right)\right)\exp\left(\frac{1}{2}\sum^{2}_{\lambda,\mu=1}\frac{J_{\lambda,\mu}}{\sqrt{N_{\lambda}N_{\mu}}}\,S_{\lambda}\left(f\left(x\right)\right)S_{\mu}\left(f\left(x\right)\right)\right)\nonumber \\
= & \sum_{x\in\Lambda_{+}}S_{1}\left(x\right)S_{2}\left(x\right)\left[\exp\left(\frac{1}{2}\sum^{2}_{\lambda,\mu=1}\frac{J_{\lambda,\mu}}{\sqrt{N_{\lambda}N_{\mu}}}\,S_{\lambda}\left(x\right)S_{\mu}\left(x\right)\right)-\exp\left(-\frac{1}{2}\sum^{2}_{\lambda,\mu=1}\frac{J_{\lambda,\mu}}{\sqrt{N_{\lambda}N_{\mu}}}\,S_{\lambda}\left(x\right)S_{\mu}\left(x\right)\right)\right].\label{eq:sum_using_bijection}
\end{align}
Now we note that due to $x\in\Lambda_{+}$ and by definition of $\Lambda_{+}$,
as well as the assumption $J_{1,2}>0$, we have
\[
\frac{J_{\lambda,\mu}}{\sqrt{N_{\lambda}N_{\mu}}}\,S_{\lambda}\left(x\right)S_{\mu}\left(x\right)>0>-\frac{J_{\lambda,\mu}}{\sqrt{N_{\lambda}N_{\mu}}}\,S_{\lambda}\left(x\right)S_{\mu}\left(x\right)
\]
for all $\lambda,\mu\in\IN_{2}$, and therefore
\[
\exp\left(\frac{1}{2}\sum^{2}_{\lambda,\mu=1}\frac{J_{\lambda,\mu}}{\sqrt{N_{\lambda}N_{\mu}}}\,S_{\lambda}\left(x\right)S_{\mu}\left(x\right)\right)>\exp\left(-\frac{1}{2}\sum^{2}_{\lambda,\mu=1}\frac{J_{\lambda,\mu}}{\sqrt{N_{\lambda}N_{\mu}}}\,S_{\lambda}\left(x\right)S_{\mu}\left(x\right)\right).
\]
Hence, each summand in (\ref{eq:sum_using_bijection}) is positive.
This proves that $J_{1,2}>0$ implies $\corr\left(S_{1},S_{2}\right)>0$.
An analogous argument yields $J_{1,2}<0\implies\corr\left(S_{1},S_{2}\right)<0$.

This concludes the proof that $J_{1,2}=0$ follows from $\corr\left(S_{1},S_{2}\right)=0$.

\subsection{\label{subsec:Sketch}Sketch of a Proof of Proposition \ref{prop:large_pop_approx}}

The proof of Proposition \ref{prop:large_pop_approx} is done in two
parts. In the first part, we determine the asymptotic (i.e.\! large
population) behaviour of correlations between votes of the form
\[
\IE\,X_{1,i_{1}}\cdots X_{1,i_{k}}\,X_{2,j_{1}}\cdots X_{2,j_{\ell}}
\]
for arbitrary $k\in\IN_{N_{1}},\ell\in\IN_{N_{2}}$.

In the second part of the proof, these correlations are employed to
calculate the limit of the moments
\[
\IE\,\frac{S^{2}_{\lambda}}{N^{\gamma}_{\lambda}}
\]
for both groups $\lambda$ and $\gamma=1$ for the high temperature
regime and $\gamma=2$ for the low temperature regime.

\subsubsection{Correlations $\protect\IE\,X_{1,i_{1}}\cdots X_{1,i_{k}}\,X_{2,j_{1}}\cdots X_{2,j_{\ell}}$}

The Hubbard-Stratonovich transformation
\[
e^{-\mathbb{H}\left(x_{11},\ldots,x_{2N_{2}}\right)}=\frac{\sqrt{\det J}}{2\pi}\int_{\IR^{2}}\exp\left(-\frac{1}{2}\left(y,z\right)J\left(y,z\right)^{T}\right)\exp\left(\left(y,z\right)J\left(S_{1}\left(x\right),S_{2}\left(x\right)\right)^{T}\right)\,\textup{d}\left(y,z\right)
\]
allows us to express the probabilities of voting configurations from
Definition \ref{def:CWM} as integrals. In the following, let for
all $t\in\left[-1,1\right]$ $P_{t}$ be the probability measure on
the set $\left\{ -1,1\right\} $ defined by $P_{t}\left\{ 1\right\} \coloneq\frac{1+t}{2}$.
We will write $P^{m}_{t}$ for the product measure $P_{t}\otimes\cdots\otimes P_{t}$
on $\left\{ -1,1\right\} ^{m}$ for all $t\in\left[-1,1\right]$ and
all $m\in\IN$. Finally, recall the definition $N\coloneq N_{1}+N_{2}$.
Then the probabilities can be expressed as
\begin{equation}
\mathbb{P}_{J,N_{1},N_{2}}\left(X_{1,1}=x_{1,1},\ldots,X_{2,N_{2}}=x_{2,N_{2}}\right)=\frac{\int_{\IR^{2}}\exp\left(-N\left(\frac{1}{2}\left(y,z\right)J^{-1}\left(y,z\right)^{T}+\frac{N_{1}}{N}\ln\cosh\left(y\sqrt{\frac{N}{N_{1}}}\right)+\frac{N_{2}}{N}\ln\cosh\left(z\sqrt{\frac{N}{N_{2}}}\right)\right)\right)P^{m}_{t}\left(x_{1,1},\ldots,x_{2,N_{2}}\right)\,\textup{d}\left(y,z\right)}{\int_{\IR^{2}}\exp\left(-N\left(\frac{1}{2}\left(y,z\right)J^{-1}\left(y,z\right)^{T}+\frac{N_{1}}{N}\ln\cosh\left(y\sqrt{\frac{N}{N_{1}}}\right)+\frac{N_{2}}{N}\ln\cosh\left(z\sqrt{\frac{N}{N_{2}}}\right)\right)\right)\,\textup{d}\left(y,z\right)}.\label{eq:de_Finetti}
\end{equation}
Let $\left(a_{m}\right)_{m\in\IN},\left(b_{m}\right)_{m\in\IN}$ be
sequences of real numbers. We write $a_{m}\approx b_{m}$ if $\lim_{m\rightarrow\infty}\frac{a_{m}}{b_{m}}=1$
holds. Set for all $\lambda\in\IN_{2}$,
\[
\alpha_{\lambda}\coloneq\lim_{N_{1},N_{2}\rightarrow\infty}\frac{N_{\lambda}}{N},
\]
and assume both of these limits exist.

Using the representation (\ref{eq:de_Finetti}), we see the expectations
$\IE\,X_{1,i_{1}}\cdots X_{1,i_{k}}\,X_{2,j_{1}}\cdots X_{2,j_{\ell}}$
equal 0 if $k+\ell$ is odd and behave asymptotically like
\[
\IE\,X_{1,i_{1}}\cdots X_{1,i_{k}}\,X_{2,j_{1}}\cdots X_{2,j_{\ell}}\approx\frac{1}{N^{\frac{k+\ell}{2}}}\frac{1}{\alpha^{\frac{k}{2}}_{1}\alpha^{\frac{\ell}{2}}_{2}}\frac{\int_{\IR^{2}}\exp\left(-N\left(\frac{1}{2}\left(y,z\right)J^{-1}\left(y,z\right)^{T}+\alpha_{1}\ln\cosh\left(\frac{y}{\sqrt{\alpha_{1}}}\right)+\alpha_{2}\ln\cosh\left(\frac{z}{\sqrt{\alpha_{2}}}\right)\right)\right)y^{k}z^{\ell}\,\textup{d}\left(y,z\right)}{\int_{\IR^{2}}\exp\left(-N\left(\frac{1}{2}\left(y,z\right)J^{-1}\left(y,z\right)^{T}+\alpha_{1}\ln\cosh\left(\frac{y}{\sqrt{\alpha_{1}}}\right)+\alpha_{2}\ln\cosh\left(\frac{z}{\sqrt{\alpha_{2}}}\right)\right)\right)\,\textup{d}\left(y,z\right)}
\]
 if $k+\ell$ is even and $I-J>0$. If $k+\ell$ is even and $I-J\geq0$
does not hold, then
\[
\lim_{N_{1},N_{2}\rightarrow\infty}\IE\,X_{1,i_{1}}\cdots X_{1,i_{k}}\,X_{2,j_{1}}\cdots X_{2,j_{\ell}}=m^{k}_{1}m^{\ell}_{2},
\]
where $\left(m_{1},m_{2}\right)\in\IR^{2}$ is the point mentioned
in Proposition \ref{prop:large_pop_approx}.

\subsubsection{Moments $\protect\IE\,\frac{S^{2}_{\lambda}}{N^{\gamma}_{\lambda}}$}

As a second step, we write the expectations $\IE\,\frac{S^{2}_{\lambda}}{N^{\gamma}_{\lambda}}$
as weighted sums of the correlations $\IE\,X_{\lambda,i_{1}}\cdots X_{\lambda,i_{k}}$
analysed above. It turns out that in the high temperature regime,
where $\gamma=1$, only correlations $\IE\,X_{\lambda,i_{1}}\cdots X_{\lambda,i_{k}}$
which feature each index $i\in\IN_{N_{\lambda}}$ either 0, 1, or
2 times, but no more than that, contribute to the expectation $\IE\,\frac{S^{2}_{\lambda}}{N_{\lambda}}$.
Take such a multi-index $\left(i_{1},\ldots,i_{k}\right)$ and write
$p$ for the number of indices which occur once and $q$ for the number
of indices which occur twice. Of course, 
\begin{equation}
p+2q=k\label{eq:profile}
\end{equation}
 and $X^{2}_{\lambda,i}=1$ hold for any $i\in\IN_{N_{\lambda}}$.
By the above analysis,
\[
\IE\,X_{\lambda,i_{1}}\cdots X_{\lambda,i_{k}}\approx\frac{C}{N^{\frac{p}{2}}}
\]
for some constant $C$ which is independent of the group sizes. Taking
into account the number of multi-indices with this property and summing
over all possible values $p,q\in\IN$ satisfying (\ref{eq:profile})
yields the claim.

In the low temperature regime, where the exponent $\gamma$ is 2,
it turns out that a single type of multi-index $\left(i_{1},\ldots,i_{k}\right)$
contributes asymptotically to the expectation $\IE\,\frac{S^{2}_{\lambda}}{N^{2}_{\lambda}}$:
the type where all indices occur exactly once. By the above result,
for such a multi-index, we have
\[
\lim_{N_{1},N_{2}\rightarrow\infty}\IE\,X_{\lambda,i_{1}}\cdots X_{\lambda,i_{k}}=m^{k}_{\lambda},
\]
and taking into account the multiplicity of such multi-indices, we
arrive at the result
\[
\lim_{N_{1},N_{2}\rightarrow\infty}\IE\,\frac{S^{2}_{\lambda}}{N^{2}_{\lambda}}=m^{k}_{\lambda}.
\]

\newpage{}

\appendix

\section*{Appendix}

\section{Auxiliary Results}
\begin{thm}[Continuous Mapping]
\label{thm:cont_mapping}Let $\left(Y_{n}\right)_{n\in\IN}$ be a
sequence of random variables and $Y$ a random variable, each of them
taking values in some subset $A\subset\IR$, such that $Y_{n}\xrightarrow[n\rightarrow\infty]{\textup{p}}Y$,
and let $g:A\rightarrow\IR$ be a continuous function. Then
\[
g\left(Y_{n}\right)\xrightarrow[n\rightarrow\infty]{\textup{p}}g\left(Y\right).
\]
\end{thm}

\begin{proof}
See Theorem 2.3 in \cite{VanderVaart1998}.
\end{proof}

\begin{thm}[Multivariate Delta Method]
\label{thm:delta_method}Let $\left(Y_{n}\right)_{n\in\IN}$ be a
sequence of random vectors taking values in some open set $D\subset\IR^{d}$
for fixed $d\in\IN$ and assume there are $\mu\in\IR^{d},\Upsilon\in\IR^{d\times d},\Upsilon>0$
such that
\begin{equation}
\sqrt{n}\left(Y_{n}-\mu\right)\xrightarrow[n\rightarrow\infty]{\textup{d}}\mathcal{N}\left(0,\Upsilon\right)\label{eq:conv_rvs}
\end{equation}
holds. Let $f:D\rightarrow\IR^{d}$ be a continuously differentiable
function with domain $D\subset\IR$. Let $f_{i}$ be the $i$-th coordinate
function, $i\in\IN_{d}$, and $\Delta$ be the Jacobian matrix of
$f$:
\[
\Delta\left(x\right)\coloneq\left(\frac{\partial f_{i}}{\partial x_{j}}\left(x\right)\right)_{i,j\in\IN_{d}}\in\IR^{d\times d},\quad x\in D.
\]

Assume that $\det\left(\Delta\left(\mu\right)\right)\neq0$. Then
\begin{equation}
\sqrt{n}\left(f\left(Y_{n}\right)-f\left(\mu\right)\right)\xrightarrow[n\rightarrow\infty]{\textup{d}}\mathcal{N}\left(0,\Delta\left(\mu\right)\Upsilon\Delta\left(\mu\right)^{T}\right)\label{eq:conv_fn}
\end{equation}
holds.
\end{thm}

\begin{proof}
Fix a coordinate $i\in\IN_{d}$. The Taylor expansion of function
$f_{i}:D\rightarrow\IR$ is given by
\[
f_{i}\left(x\right)=f_{i}\left(\mu\right)+\sum^{d}_{j=1}\frac{\partial f_{i}}{\partial x_{j}}\left(\mu\right)\left(x_{j}-\mu_{j}\right)+\sum^{d}_{k=1}\sum^{d}_{\ell=1}h_{k,\ell}\left(x\right)\left(x_{k}-\mu_{k}\right)\left(x_{\ell}-\mu_{\ell}\right),
\]
where, for each $k,\ell\in\IN_{d}$, $h_{k,\ell}:D\rightarrow\IR$
is a function which satisfies $\lim_{x\rightarrow\mu}h_{k,\ell}\left(x\right)=0$.
Therefore,
\[
\sqrt{n}\left(f_{i}\left(Y_{n}\right)-f_{i}\left(\mu\right)\right)=\sum^{d}_{j=1}\frac{\partial f_{i}}{\partial x_{j}}\left(\mu\right)\sqrt{n}\left(Y_{n,j}-\mu_{j}\right)+\sum^{d}_{k=1}\sum^{d}_{\ell=1}h_{k,\ell}\left(x\right)\sqrt{n}\left(Y_{n,k}-\mu_{k}\right)\left(Y_{n,\ell}-\mu_{\ell}\right).
\]
The second summand on the right hand side above converges to 0 in
probability due to $\lim_{x\rightarrow\mu}h_{k,\ell}\left(x\right)=0$.
By (\ref{eq:conv_rvs}), $\sqrt{n}\left(Y_{n}-\mu\right)\xrightarrow[n\rightarrow\infty]{\textup{d}}\mathcal{N}\left(0,\Upsilon\right)$
holds. It follows that the first summand has the limiting distribution
\begin{equation}
\sum^{d}_{j=1}\frac{\partial f_{i}}{\partial x_{j}}\left(\mu\right)\sqrt{n}\left(Y_{n,j}-\mu_{j}\right)\xrightarrow[n\rightarrow\infty]{\textup{d}}\mathcal{N}\left(0,\nabla f_{i}\left(\mu\right)^{T}\Upsilon\,\nabla f_{i}\left(\mu\right)\right),\label{eq:conv_marginals}
\end{equation}
where
\[
\nabla f_{i}\left(\mu\right)\coloneq\left(\frac{\partial f_{i}}{\partial x_{j}}\left(\mu\right)\right)_{j\in\IN_{d}}\in\IR^{d}
\]
is the gradient of $f_{i}$ at $\mu$. We note that
\begin{equation}
\Delta\left(\mu\right)=\left(\begin{array}{c}
\nabla f_{1}\left(\mu\right)^{T}\\
\vdots\\
\nabla f_{d}\left(\mu\right)^{T}
\end{array}\right).\label{eq:Delta_grad}
\end{equation}
Set
\[
U_{n,i}\coloneq\sum^{d}_{j=1}\frac{\partial f_{i}}{\partial x_{j}}\left(\mu\right)\sqrt{n}\left(Y_{n,j}-\mu_{j}\right),\quad n\in\IN,i\in\IN_{d}.
\]

Let $a_{1},\ldots,a_{d}\in\IR$ and consider the sequences of random
variables $\left(U_{n,i}\right)_{n\in\IN}$, $i\in\IN_{d}$. By (\ref{eq:conv_marginals}),
we have
\[
\sum^{d}_{i=1}a_{i}U_{n,i}\xrightarrow[n\rightarrow\infty]{\textup{d}}\mathcal{N}\left(0,\sum^{d}_{i=1}a^{2}_{i}\,\nabla f_{i}\left(\mu\right)^{T}\Upsilon\,\nabla f_{i}\left(\mu\right)\right).
\]
Since by the last display any linear combination of the random variables
$U_{n,i}$ converges to a univariate normal distribution, the random
vectors $\left(U_{n,1},\ldots,U_{n,d}\right)$ converge to a multivariate
normal distribution. Taking into account (\ref{eq:Delta_grad}), we
conclude that (\ref{eq:conv_fn}) holds.
\end{proof}

\section*{Acknowledgements}

M.\! B.\! is a Fellow and G.\! T.\! is a Candidate of the Sistema
Nacional de Investigadoras e Investigadores. G.\! T.\! was supported
by a Secihti (formerly Conahcyt) postdoctoral fellowship.

\bibliographystyle{plain}

\end{document}